\documentclass[aps,prl,twocolumn,showpacs,preprintnumbers]{revtex4-1}

\usepackage{psfrag,graphicx}
\usepackage{dcolumn}
\usepackage{amsmath,amssymb}
\usepackage{bm}
\usepackage{amsfonts,amssymb,amsmath}        % for math symbols.
\usepackage{epstopdf}
\usepackage{amsthm}

\newcommand{\beq}{\begin{equation}}
\newcommand{\eeq}{\end{equation}}
\newcommand{\bea}{\begin{eqnarray}}
\newcommand{\eea}{\end{eqnarray}}

\newtheorem{lemma}{Lemma}
\newtheorem{definition}{Property}

\begin{document}

\title{Active error correction for Abelian and non-Abelian anyons}
\author{James R.~Wootton and Adrian Hutter}
\affiliation{Department of Physics, University of Basel, Klingelbergstrasse 82, CH-4056 Basel, Switzerland}

\date{\today}

\begin{abstract}

We consider a class of decoding algorithms that are applicable to error correction for both Abelian and non-Abelian anyons. This class includes multiple algorithms that have recently attracted attention, including the Bravyi-Haah RG decoder. They are applied to both the problem of single shot error correction (with perfect syndrome measurements) and that of active error correction (with noisy syndrome measurements). For Abelian models we provide a threshold proof in both cases, showing that there is a finite noise threshold under which errors can be arbitrarily suppressed when any decoder in this class is used. For non-Abelian models such a proof is found for the single shot case. The means by which decoding may be performed for active error correction of non-Abelian anyons is studied in detail. Differences with the Abelian case are discussed.

\end{abstract}

\maketitle

\section{Introduction}

The possibility of using anyonic quasiparticles for quantum computation has inspired a great deal of research \cite{pachos:book}. This is due in part to the idea of `topological protection', which promises inherent fault-tolerance for anyonic systems. Nevertheless, this protection still comes at a price. Without active error correction \cite{wootton:error}, or additional passive protection \cite{ben:rev}, the fault tolerance will fail after a system size independent lifetime \cite{alicki:07,wootton:error,fabio:15}. 
Though one can hope to extend this through means such as lowering temperature, such an approach is not consistent with the scalability required for quantum computation. 
It is therefore important to study how error correction may be performed in anyonic systems.	

For Abelian anyons the problem of error correction has been, and continues to be, studied in great detail \cite{fowler,bravyi:14,anwar:14,hutter:improved}. Many good decoding algorithms are known, and proofs that these allow exponential suppression of logical errors below a finite noise threshold have been found in multiple cases \cite{bravyi:13,fowler:12}. For non-Abelian anyons, however, this study is in its infancy \cite{wootton:error,brell:14,hutter:improved,burton:15}. The only case considered so far is a `single shot' scenario. This assumes an initial burst of noise, with all measurements and manipulations performed perfectly thereafter. The more realistic problem of dealing with continuously occurring noise through active error correction has hardly been considered \cite{wootton:error}.

In this work we specifically consider a certain class of decoders. These can correspond to quite different methods, and yet have shared properties that allow them to be studied collectively. Examples of such decoders have recently been considered for multiple problems in Abelian and non-Abelian error correction \cite{sarvepalli:12,bravyi:13,wootton:simple,wootton:error,brell:14,anwar:14,hutter:improved,watson:15,brown:15} We provide a general proof of a finite noise threshold for these decoders, applicable to single shot error correction for Abelian and non-Abelian anyons, as well as active error for Abelian anyons.

For active error correction of non-Abelian anyons, we study the way in which syndrome measurements must be interpreted in order for the decoders to be applied. Differences between the Abelian and non-Abelian cases are found and discussed. Specifically, it is shown that these prevent the proof used for Abelian active correction from being adapted to the non-Abelian case.

\section{Definitions}

\subsection{Code and Syndrome Lattice}

The proof concerns error correcting codes defined on a two-dimensional lattice with quasilocal syndrome operators, such that their eigenspaces can be identified with anyonic occupations.

For concreteness we consider models based on a two-dimensional $L \times L$ lattice which we call the `code lattice'. Anyons are associated with plaquettes, $P$, and the errors that affect a pair of neighbouring plaquettes are associated with the edge between them. The errors for each edge are assumed to act independently for analytical convenience.

A model of this form can be constructed for any anyon model. This framework may therefore be used to study general properties of anyonic decoding, when there is no need to specify the actual physical system used. They have especially been used to construct toy models for non-Abelian anyons. \cite{brell:14,burton:15}. 

Error correction first requires the anyonic occupancy of the plaquettes to be measured. If the code is Abelian and the syndrome measurements are without noise, these results provide sufficient information for error correction to be performed. The input to the decoder in this case is therefore a two-dimensional syndrome composed of these measurement results. This is the single shot case for Abelian anyons.

For non-Abelian anyons the single shot case is more contrived. As well as perfect syndrome measurements, a lack of any noise while anyon fusions are performed must also be assumed. The problem therefore has little physical relevance, beyond providing a first glimpse into non-Abelian decoding. As for the Abelian case, the syndrome given to the decoder is two-dimensional \cite{wootton:error,brell:14,hutter:improved,burton:15}.

When measurement results are noisy, a single measurement of each plaquette is no longer sufficient for good error correction. Instead, each syndrome operator must be measured periodically. Let us use $T$ to denote the total number of measurement rounds. The measurement results at each time step can then be used to generate a three-dimensional syndrome, of size $L$ in each spacial direction and $T$ in the time direction. 

Let us now construct a lattice on which the syndrome can be analyzed, which we call the `syndrome lattice'. Consider the code lattice stacked upon itself $T$ times to form a three-dimensional structure. We then define a set of points labelled $(P,t)$ to lie directly between the copies of the plaquette $P$ at timeslices $t$ and $t+1$. These points are taken to be the vertices of the syndrome lattice. So-called `time-like' edges are placed between each pair of vertices $(P,t)$ and $(P,t+1)$.  `Space-like' edges are placed between each $(P,t)$ and $(P',t)$ for neighbouring plaquettes $P$ and $P'$. This generalizes a well-known procedure for surface codes \cite{dennis}.

A syndrome value is assigned to each vertex of the syndrome lattice. These values reflect the difference between the measured anyon occupancy for the plaquette at these times. The exact details of how this is done depends on whether the anyons are Abelian or non-Abelian, and so will be specified in their respective sections.

Changes in anyon occupancy, as detected by this syndrome, are caused by errors. An error on the code between times $t$ and $t+1$ that changes the anyon occupancies of $P$ and $P'$ is associated with the space-like edge between $(P,t)$ and $(P',t)$. A measurement error for a plaquette $P$ during the round $t$ is associated with the time-like edge between $(P,t-1)$ and $(P,t)$.

We assume a toric or planar variant of the topological codes, for which logical information is stored within the degenerate vacuum states of the anyons. For this case, the code distance is $L$. Our results also apply to other means of storing logical information, such as holes \cite{wootton:rev}, defects \cite{bombin:twist,wootton_family} or using non-Abelian anyons themselves \cite{pachos:book}. In these cases the code distance $L'<L$ reflects the distance between these structures. Our results apply straightforwardly to these cases, with the simple substitution $L \rightarrow L'$.

\subsection{Code Lattice for Quantum Double Models}

For concreteness let us consider the quantum double models \cite{double}, a specific class of topological codes based on explicit spin lattice models that can realize both Abelian and no-Abelian anyons. These models are based on a two-dimensional lattice, however, it does not correspond exactly to the code lattice as defined above. This is because syndrome operators are defined on both the vertices and the plaquettes.

For Abelian models, the set of anyons living on plaquettes and vertices are independent of one another in terms of their creation and fusion. They may therefore be decoded independently. One could therefore consider a two independent code lattices: one for which the plaquettes correspond to quantum double plaquettes, and one for which they correspond to quantum double vertices.

This property does not hold for general quantum double models, though. In the non-Abelian case, it is possible for plaquette anyons to fuse to vertex ones. They are therefore no longer independent. We must therefore reinterpret these models in order to find a code lattice in the simple form we desire.

Let us consider a quantum double model defined on a square lattice. A spin is associated with each edge of this lattice. To each plaquette, $p$, we assign the vertex $v$ to its top right. This results in six spins around each combined $(p,v)$. A hexagonal lattice can then be drawn such that these spins lie on the vertices. This will be used as the code lattice. Each combined $(p,v)$ from the original lattice then corresponds to a single plaquette, $P$, in the code lattice. Each plaquette, $P$, has both kinds of stabilizer operator associated with it, and so can hold all possible kinds of anyon in the model. This lattice is shown in Fig. \ref{fig:qd}.

Note that each pair of neighbouring plaquettes, $P$ and $P'$, share two spins. For one of these, the only errors that would affect $P$ and $P'$ are those that affect flux anyons. For the other, only the errors for charge anyons affect $P$ and $P'$. When using the code lattice, errors are associated with edges rather than vertices. These errors from different spins are therefore associated with the single edge that lies between $P$ and $P'$. The independence of errors on each edge of the code lattice therefore requires not only that errors on each spin are independent of each other, but also independence of flux and charge errors on the same spin. We therefore assume such noise when considering quantum double models.

\begin{figure}[t]
\begin{center}
{\includegraphics[width=8cm]{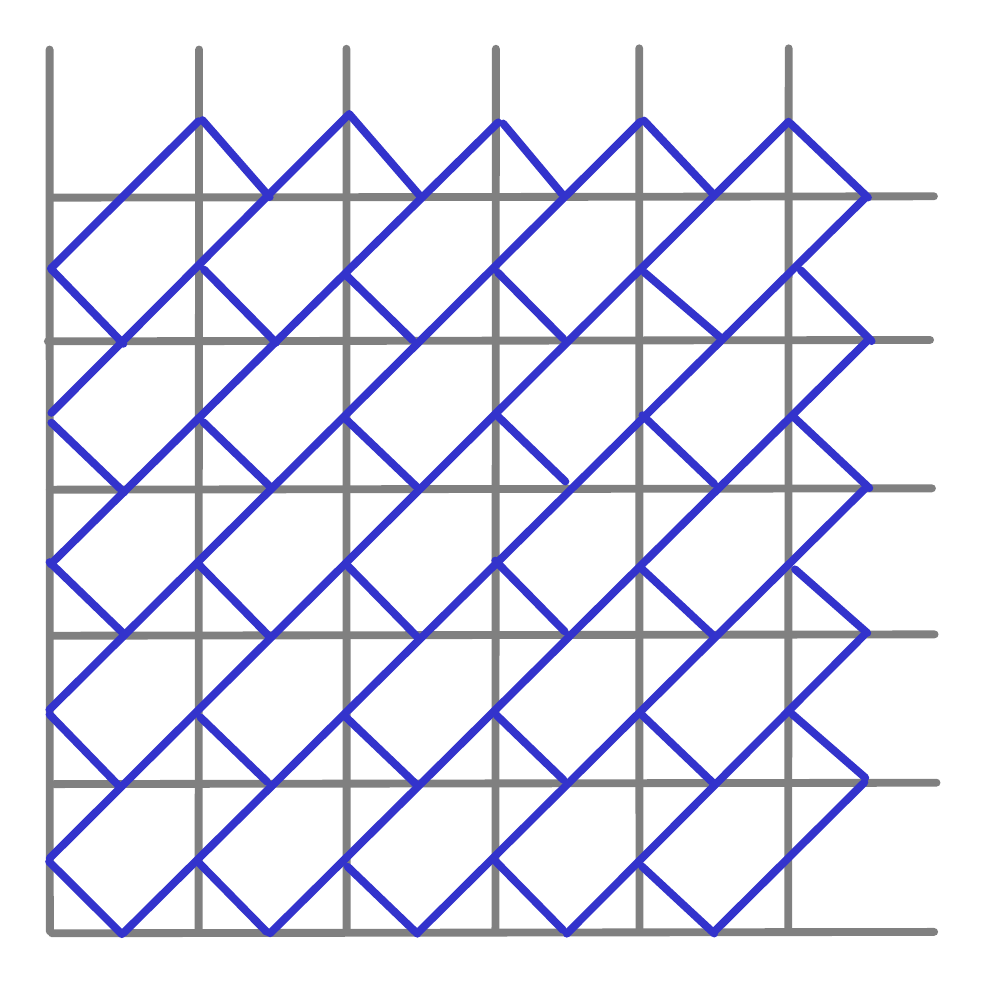}}
\caption{\label{fig:qd} Square lattice on which a quantum double model is defined is shown in grey. The corresponding hexagonal code lattice is shown in blue. Spins are located on the edges of the former, and vertices of the latter.}
\end{center}
\end{figure}

\subsection{Clusters and Chunks}

We use $E$ to denote the set of errors that occur, including both spin and measurement errors. This is therefore a set of edges on the syndrome lattice. We use $S=S(E)$ to denote the corresponding set of non-trivial syndrome elements, which is a set of syndrome lattice vertices. Here $S(E)$ refers to 3D syndrome based on changes in measurement results, rather than the measured anyon occupancies themselves.

Any subset of the vertices of the syndrome lattice is called a cluster. Typically the clusters considered are those for which all vertices are occupied by an element of $S$. When this is not true, the cluster is called a  vertex cluster.

Any subset of the edges of the syndrome lattice is called a chunk. We will only consider chunks made up of edges associated with an element of $E$. Chunks are therefore also subsets of $E$.

The decoders we consider use the distance between non-trivial syndrome elements to determine how to best correct the errors that caused them. A sensible choice for the distance $d(k,k')$ between two vertices $k$ and $k'$ of the syndrome lattice is therefore the minimum number of edges required to connect them. This will be the metric that we primarily consider. However, any metric for which all distances are integers could also be used.

Two clusters, $C$ and $C'$, are said to overlap if there exists $k_1, k_2 \in C$ and a $k' \in C'$ such that $d(k_1,k') \leq d(k_1,k_2)$. The cluster $C'$ is said to be inside $C$ if the above is true for all $k' \in C'$.

The width of a cluster, $C$, is defined to be the distance between its extremal points,
\beq
W(C) = \max_{k,k' \in C} d(k,k').
\eeq
The distance between two overlapping clusters is defined to be zero. For non-overlapping clusters it is the distance between their closest points
\beq
d(C,C') = \min_{k\in C,k' \in C'} d(k,k').
\eeq
Note that this distance does not satisfy the triangle inequality. If clusters $C$ and $C'$ are separated by a finite distance, but both overlap with a cluster $C''$, then $d(C,C')>d(C,C'')+d(C'',C')$. This fact will not present a problem for the proof, but should be kept in mind.

To define widths and distances for a chunk $\varepsilon$, we consider the vertex cluster $C(\varepsilon)$ composed of all vertices adjacent to elements of the chunk. The width of the chunk is then defined to be the width of $C(\varepsilon)$. The distance between two chunks is the distance between their corresponding vertex clusters.

A cluster, $C$, is called neutral if there exists a chunk $\epsilon$ such that $C = C(\epsilon)$. Note that this $\epsilon$ need not be present in the actual error, $E$. The neutrality of $C$ simple means that it is possible for it to have been created by some set of errors without otherwise affecting the syndrome. This means that it is also possible to correct the cluster independently of the rest of the syndrome. Finding neutral clusters is therefore an important part of decoding.

We call a chunk, $\varepsilon$, `disconnected' if it generates its own syndrome cluster that does not depend on the rest of $E$, i.e.
\beq
S(E \setminus \varepsilon) \cup S(\varepsilon) = S, \,\,\, S(E \setminus \varepsilon) \cap S(\varepsilon) = \emptyset.
\eeq
A sufficient condition for this is clearly that the vertex clusters $C(\varepsilon)$ and $C(E \setminus \varepsilon)$ are disjoint, and so $d(\varepsilon, E \setminus \varepsilon) \geq 1$. Note that the syndrome cluster $S(\varepsilon)$ created by a disconnected chunk will be neutral by definition.

\subsection{Error Model}

To continue with our analysis, the error model must be specified. As stated earlier, we assume that the errors associated with each edge of the syndrome lattice occur with an independent probability distribution. This requires there to be no correlations between errors on different spins, and no correlations between charge and flux errors on each spin for quantum double models. However, it will be allowed for the error probability to depend on the occupancy of the two plaquettes adjacent each the spin. This occurs when there is an energy gap for anyon creation, for example.

For the measurement errors, we consider a model in which the measurement simply reports an incorrect value. This is the simplest model that allows us to study the nature of decoding when measurement errors are present, and is often used for benchmarking. More realistically one should consider all elements of the process, such as a quantum circuit, performing the measurements and include realistic errors in each. However, since this will be very specific to each individual code, it is not compatible with our general approach.

We will quantify the strength of the noise using an upper bound on the probability that any kind of error will occur. Let us first consider this for the case of a charge error occurring on a spin during the time between two measurement rounds. Using $j$ to denote a possible error type for a spin and $k_P$ and $k_{P'}$ to denote the occupancies of the adjacent plaquettes, we define $p_z$ to be
\beq
p_z = \max_{k_P, k_{P'}} \sum_j {\rm Prob} (j | k_P, k_{P'}).
\eeq
It is therefore the total probability that an error of any kind will occur, for the anyon occupancies for which an error is most likely. The corresponding probability, $p_x$, for flux errors is defined in the same way. The maximum probability for any kind of error associated with any space-like edge of the syndrome lattice is then $p_s = p_x + p_z - p_x p_z$.

The probability for measurement errors, and hence time-like edges is
\beq
p_m = \max_{k} \sum_{j \neq k} {\rm Prob} (j | k).
\eeq
Here $j$ denotes a possible outcome reported by the measurement, and $k$ denotes the true value. The probability $p_m$ is therefore the total probability that the measurement reports any wrong value, for the true value for which an error is most likely.

We now combine this error rates into a single value $p = \max (p_s,p_m)$. This is an upper bound for the error probability for any kind of error event associated with any edge of the syndrome lattice.

\subsection{Chunk decomposition}

Let us now follow \cite{bravyi:13} by using the concept of level-$n$ chunks. The definition of these depends upon a constant $Q$  whose value can be chosen arbitrarily.

A level-$0$ chunk is defined to be a single error. A level-$n$ chunk is a union of two disjoint level-$(n-1)$ chunks such that the width is at most $Q^n$. A level-$n$ chunk therefore contains exactly $2^n$ errors.

We use $E_n$ to denote the union of all possible level-$n$ chunks. Note that this is not a disjoint union: the same errors could be involved in multiple possible level-$n$ chunks. Clearly $E = E_0$, and
\beq
E_0 \supseteq E_1 \supseteq \ldots \supseteq E_m.
\eeq
Here $m$ is the highest level for which $E_m \neq \emptyset$, given the error $E$.

It is useful to reflect upon the meaning of the sets $E_n$. For the following two lists, `within a distance' is used to mean `such that the union has a width no greater than'.
\begin{itemize}
\item $E_0$ is the set of all errors.
\item $E_1$ is the set of all errors within a distance $Q$ of another.
\item $E_2$ is the set of all errors within a distance $Q$ of another for which there is another such pair within a distance $Q^2$.
\item $E_3$ is the set of all errors within a distance $  Q$ of another, for which there is another such pair within a distance $  Q^2$, for which there is another such quadruple with a distance $  Q^3$.
\item $\ldots$
\end{itemize}

Using the sets $E_n$, we define the sets $F_n = E_n \setminus E_{n+1}$. These are the errors that form part of $E_n$ but not $E_{n+1}$, so
\beq
E = F_0 \cup F_1 \cup \ldots \cup F_m.
\eeq
This is a disjoint union, which is called the `chunk decomposition' of $E$. Again it is useful to reflect upon the meanings of these sets.
\begin{itemize}
\item $F_0$ is the set of all errors further than $  Q$ from any other.
\item $F_1$ is the set of all errors for within a distance $  Q$ of another for which there is no other such pair within a distance $  Q^2$.
\item $F_2$ is the set of all errors for within a distance $  Q$ of another for which there is another such pair within a distance $  Q^2$ but no other such quadruple within a distance $  Q^3$.
\item $\ldots$
\end{itemize}

\section{Greedy HDRG decoders}

Decoders based on greedy algorithms, in which syndrome elements attempt to neutralize themselves with near neighbours without considering the rest of the syndrome, will typically lead to a logical error rate that decays exponentially with $L^\beta$ for $\beta<1$. This is less than the optimal $\beta=1$ scaling, and is due to greedy algorithms being fooled by Cantor like error chains \cite{wootton:simple,hutter:improved}. However, such algorithms do typically have nice properties for analytical treatment. Specifically, any neutral cluster that is sufficiently far from the rest of the syndrome will typically be corrected independently of the rest.

Let us make this more rigorous. Decoders take a syndrome $S$ as an input and yield a correction operator $E_c(S)$ as an output. A cluster $C$ is called `independent' if
\beq
E_c(S) = E_c(C) \times E_c(S \setminus C).
\eeq 
Note that here $E_c(S)$ is an operator acting on the Hilbert space of the code, and so the multiplication should be interpreted accordingly.

A disconnected chunk is similarly called independent if its syndrome cluster $S(\varepsilon)$ is independent. Note that since an independent cluster is disconnected by definition, and a disconnected cluster is neutral by definition, independent clusters will always be neutral.

Greedy HDRG decoders are then defined such that the following two properties hold.

\begin{definition} \label{prop:1}

For an independent chunk of width $W$, the width of the correction operator is no greater than $W+O(1)$.

\end{definition}

\begin{definition} \label{prop:2}

Any chunk $\varepsilon$ of width $W$ is independent as long as there is a distance of greater than $\lambda \, W/2$ from it to $E \setminus \varepsilon$. Here $\lambda$ is a decoder dependent constant.

\end{definition}

Such decoders have recently been considered in References \cite{bravyi:13,wootton:simple,wootton:error,brell:14,anwar:14,hutter:improved,watson:15}.

\section{Threshold proof for greedy decoders}

A decoder is only truly useful for fault-tolerance if there exists a threshold $p_c$ such that the probability of a logical error vanishes for $p<p_c$ and $L \rightarrow \infty$. The nature of the decay with $L$ is also important. Here we prove bounds for these for any decoder of the type described above. Here we formulate the proof in a way that can be applied to both the single shot and active error correction problems. The only difference is the dimension of the syndrome lattice, with $D=2$ for the former case and $D=3$ for the latter.

For the proof we require a value of $Q$ such that the following holds true. For any $u \in F_n$, let $\varepsilon$ denote the chunk composed of all errors no further than $Q^n$ from $u$. For any $v \in E_n$ we then require that either:
\begin{itemize}
\item $v \in \varepsilon$;
\item $v$ is further than $\lambda Q^n$ from any element of $\varepsilon$.
\end{itemize}
A necessary and sufficient condition for the former is $d(u,v) \leq Q^n$. For the latter, the condition $d(u,v) > (\lambda + 1) Q^n$ is sufficient. We will define $Q$ such that both of these will always hold.

\begin{lemma} \label{lemma:1}

For any $u \in F_n$ there is no $v \in E_n$ that satisfies
\beq \label{eqn:distcond}
Q^n < d(u,v) \leq (\lambda + 1) Q^n .
\eeq
as long as $Q \geq \lambda+3$.

\end{lemma}

\begin{proof}

Let us consider a pair of errors $u,v \in E_n$ that \emph{do} satisfy Eq.~\ref{eqn:distcond}. Since both errors  are in $E_n$, both are contained within level-$n$ chunks. Let us denote these $C_u$ and $C_v$, respectively. Since chunks must have a width no greater than $Q^n$ by definition, the condition that $d(u,v) > Q^n$ means that $C_u$ and $C_v$ must be disjoint chunks.

Despite the non-applicability of the triangle inequality, the width of the combined chunk $C_u \cup C_v$ will clearly satisfy
\beq
W(C_u \cup C_v) \leq W(C_u) + W(C_v) + d(u,v).
\eeq
Again using the width restriction, as well as the condition that $d(u,v) \leq (\lambda + 1) Q^n$, we find
\beq
W(C_u \cup C_v) \leq 2 Q^n + (\lambda + 1) Q^n .
\eeq
The combined chunk will form a valid level-$(n+1)$ chunk if its width is no greater than $Q^{n+1}$. Clearly this will be satisfied for all $Q \geq \lambda + 3$. Since both $u$ and $v$ will be contained within a level-$(n+1)$ chunk in this case, neither will be an element of $F_n$. It therefore follows that, whenever either $u$ or $v$ is an element of $F_n$, Eq.~\ref{eqn:distcond} cannot hold.

\end{proof}

With the chunk decomposition so defined, it can allow us to easily identify independent chunks of errors.

\begin{lemma} \label{lemma:2}

For any error $u$ and the corresponding set $F_n$, let $\varepsilon$ denote the chunk composed of all errors no further than $Q^n$ from $u$. For $Q \geq \lambda+3$, all such $\varepsilon$ will be independent for any decoder that satisfies Property \ref{prop:2}.

\end{lemma}

\begin{proof}

Clearly the maximum width of any such chunk is $W \leq 2Q^n$. By Lemma \ref{lemma:1} we know that such chunks are a distance of at least $\lambda Q^n \geq \lambda W/2 $ from any other element of $E_n$. Any errors within this distance must therefore be elements the sets $E_{n'}$ for lower levels $n'<n$.

There are no lower levels than $E_0$, so let us proceed by induction. Any such $\varepsilon$ based around a $u \in F_0$ will have no errors within a distance $\lambda W$, and so will be independent for any decoder that satisfies Property \ref{prop:2}. Since all errors in $F_0$ will be corrected independently, the decoder will treat the remaining errors in the same way as if the original error was $E \setminus F_0 = E_1$.

Similarly, none of the remaining errors $E_1$ will be within a distance $\lambda W/2$ of any $\epsilon$ based around a $u \in F_1$. All errors in $F_1$ are therefore also corrected independently, and the decoder act on the remaining errors as if the original error was $E \setminus F_0 \setminus F_1 = E_2$. Continuing this process, we find that all $\varepsilon(u)$ for $u \in F_n$ are independent chunks, as required.

\end{proof}

The chunk decomposition therefore forms a decomposition of the errors into independently correctable chunks. This allows us to identify those errors that will cause the decoder to fail.

\begin{lemma}

A necessary condition for a logical error is that the highest level in the chunk decomposition satisfies $m \geq \gamma \log(L/2)/\log(Q)$.

\end{lemma}

\begin{proof}

By Property \ref{prop:1}, any independent chunk is neutral and so can be corrected by a operator whose width is (essentially) no greater than that of the chunk. For an independent chunk to cause a logical error, its correction operator must have a width as large as $L$, the code distance. This requires independent chunks with width $W=2Q^n \geq L$. The lowest value of $m$ for which these can occur is $m \geq \log(L/2)/\log(Q)$, giving the required result.
\end{proof}

Now we can analyse the probability that a level-$m$ chunk arises, for a given $m$. For this, consider the $D$-dimensional boxes $\Sigma_n$ and $\Sigma_n^+$, centred on the same point. The former is sufficiently large to contain a chunk of width $Q^n$, and the latter can contain one of width $3Q^n$. Using these, consider the following events.
\begin{itemize}

\item $A_n$ : $\Sigma_n$ contains at least part of a level-$n$ chunk.

\item $B_n$ : $\Sigma_n^+$ contains a level-$n$ chunk.

\item $C_n$ : $\Sigma_n^+$ contains a level-$(n-1)$ chunk.

\end{itemize}

Due to the width restriction on chunks, $A_n$ is a sufficient condition for $B_n$. Their probabilities are therefore related by
\beq
P_{B_n} \geq P_{A_n} .
\eeq
Note that $B_n$ requires $\Sigma_n^+$ to contain two disjoint level-$(n-1)$ chunks. Two independent occurrences of event $C_n$ are a necessary condition for this, so
\beq
P_{B_n} \leq P_{C_n}^2.
\eeq
Note that $\Sigma_n^+$ is composed of $q=(3Q)^{D}$ disjoint boxes $\Sigma_{n-1}$. The event $A_{n-1}$ on at least one $\Sigma_{n-1}$ is therefore a necessary condition for the event $C_n$ on $\Sigma_n^+$, and so
\beq
P_{C_n} \leq q P_{A_{n-1}}.
\eeq
Putting this all together, we obtain the recursive relation
\beq
P_{A_n} \leq ( q P_{A_{n-1}} )^2.
\eeq
The event $P_{A_0}$ is that of a single error, which is upper bounded by $p$ by definition. Repeatedly applying the recursive relation then allows us to express the probability of a level-$m$ chunk in terms of $p$,
\beq
P_{A_m} \leq (q^2 p)^{2^m}.
\eeq

For the single shot case, as well as that of active error correction when $T=L$, we consider a syndrome of size $L^D$. Since a chunk of size $m \geq \log(L/2)/\log(Q)$ is a necessary condition for a logical error, the logical error rate is upper bounded by,
\beq \label{eq:log}
P \leq \, \left[(3Q)^{2 D} p \right]^{(L/2)^\beta}, \,\,\,\, \beta \geq \frac{1}{\log_2 Q}.
\eeq
Note that $P$ decays exponentially in $(L/2)^\beta$ when
\beq
p < (3Q)^{- 2 D}.
\eeq
This therefore gives a lower bound on the threshold, $p_c$, for this decoding problem and decoder.

For $T>L$, a necessary condition for a logical error is for a chunk of size $m \geq \log(L/2)/\log(Q)$ to intersect at least one of the $(T/L)$ boxes of size $L \times L \times L$ that make up the $L \times L \times T$ syndrome. The probability for this will clearly share the exponential factor of Eq.~\ref{eq:log}. The same threshold applies therefore applies for arbitrary $T>L$.

The combined bounds are then
\beq \label{eq:thresh}
p_c \geq (3Q)^{-2 D}, \,\,\, \beta \geq \frac{1}{\log_2 Q}
\eeq
Note that the threshold and the exponent $\beta$ both depend on $Q$.

\section{Application to Abelian models}

Let us now consider the specific case of a quantum double model is based on an Abelian group. The results of the syndrome measurements can therefore be interpreted in terms of Abelian anyons \cite{double}. Any finite Abelian group is a product of cyclic groups $Z_d$. The resulting quantum double model is then the corresponding tensor product of the models based on each of these factors. As such we restrict to cyclic groups without loss of generality.

The way to analyse changes in the measured syndrome in order to perform active error correction is well known for these models ~\cite{duclos:14,hutter:improved,watson:15}. Nevertheless, we explain it here in detail.

We specifically consider error correction for the case in which the logical information is being stored in the code, and not manipulated. As such, though syndrome readout is being performed constantly, error correction can be delayed until readout.

The quantum double model $D(Z_n)$ has $n^2$ different species of anyons that can live in each of the plaquettes, $P$. These can be denoted $e_g m_h$ for $g,h \in Z_d$, and have the fusion rules
\beq
e_g m_h \times e_{g'} m_{h'} = e_{g+g'} m_{h+h'}. 
\eeq
Here addition is taken modulo $n$. The anyon $e_0 m_0$ is identified with the vacuum.

Without errors, the syndrome measurements would never change. As such, changes in the measurement results are signatures of errors. Such changes will not necessarily occur adjacent to every error. Instead, they are found at the endpoints of error chains. The type of error chain that can terminate at any syndrome change depends on the nature of the change.

In order to correct the errors we must consider what error chains are consistent with the syndrome measurements. We therefore need to determine exactly what syndrome changes have occurred, where they occurred and when. The details of the changes can then be placed on the three dimensional syndrome lattice (with two dimensions for space and one for time). If the outcome of the measurement of $P$ at $t-1$ is $e_g m_h$, and that of the same plaquette at $t$ is $e_{g'} m_{h'}$, the corresponding vertex $(P,t)$ of the syndrome lattice is assigned the value $e_{g'-g} m_{h'-h}$. This gives the trivial value $e_0 m_0$ if the two results are the same, signalling that no error has been detected. Otherwise a non-trivial syndrome element is present at $(P,t)$. We refer to these as `defects'. Note the syndrome lattice of defects contains the same information as the list of all measurement results. However, it presents the information in a form that is more convenient for analysing error chains.

The first step towards determining a likely error chain for an Abelian model could be to choose defects that are likely joined by an error chain, and draw an error chain between them. Let us use $e_{g_1} m_{h_1}$ to denote the type of one of these defects, and $e_{g_2} m_{h_2}$ to denote that of the other. If
\beq
g_1 + g_2 = m_1 + m_2 = 0 \,\,\, \mod n,
\eeq
this pair of defects can be said to be neutral. This means that no further error chains are required to explain this pair of syndrome changes. Otherwise the pair is non-neutral. In this case we can cease to regard the two points as being defects individually. Instead they collectively make up a single defect, along with the error chain that connects them. This must be connected with error chains to further defects in order to be resolved. Once such a cluster of syndrome changes, $C$, satisfies
\beq
\sum_j g_j = \sum_j h_j = 0 \,\,\, \mod n,
\eeq

where $e_{g_j} m_{h_j}$ is the syndrome value for each $j \in C$, it can be said to be neutral. The set of error chains connecting the syndrome changes is then sufficient to explain their presence without being dependent on any other part of the syndrome.

Correction of a neutral cluster is done by moving its defects together. These obey the same fusion rules as the anyons, and so the effect of moving all the defects together is to annihilate them. Moving a defect along a time-like interval implies that the measurement results along the interval were incorrect. The movement is done by correcting the results by changing their values. Moving along a space-like interval implies that errors occurred on the spins along the interval, and is done by applying the inverse of the corresponding errors. In both cases the required operations are applied to edges of the syndrome lattice.

Note that the defects created by a chunk $\varepsilon$ are always inside $C(\varepsilon)$. The movement of defects by a decoder can always be implemented such that they always remain inside $C(\varepsilon)$. As such, we can always assume that the decoder satisfies Property \ref{prop:1}.

\subsection{Bravyi-Haah and ABCD Decoders}

The Bravyi-Haah decoder \cite{bravyi:13,watson:15} runs an iterative process to find neutral clusters. In the $n$th iteration, all defects within a distance $2^n$ of each other are placed in the same cluster. Neutral clusters are then identified and removed from the syndrome. If defects remain, the process is repeated for $n+1$.

The minimum $n$ required to cover an independent chunk $\varepsilon$ of width $W$ is $n = \left \lceil \log_2 W \right \rceil$. In order for no defects from $E \setminus \varepsilon$ to be included within the same cluster, they must be more than a distance $2^n$ away.

Note that the same $n$ covers distances from $2^{n-1}+1$ to $2^n$. The minimum $W$ that requires a distance of $2^n$ to other chunks is therefore $2^{n-1}+1$. This means that a chunk of width $W$ requires a distance of greater than $2(W-1)$ to be independent. This decoder is therefore a greedy HDRG decoder with $\lambda = 4$.

The ABCB decoder \cite{anwar:14} is based on the same principle as that above. However, the distance used for iteration $n$ is simply $n$ rather than $2^n$. As such $\lambda = 2$ for this decoder.

\subsection{Expanding Diamonds}

The expanding diamonds decoder \cite{dennis:thesis,wootton:simple} is based on a similar iterative process to the above. Initially, the syndrome is decomposed into clusters such that each defect corresponds to its own cluster. During iteration $n$, each cluster checks whether another exists at a distance $n$ away. If so, the clusters can be paired. Each pair is removed from the syndrome if neutral. Once no more pairs are possible for the distance $n$, the distance $n+1$ is considered.

The largest distance required for all defects within an independent cluster to see each other is its width. In order for none to see any defects outside the independent cluster before they become neutral, the distance to other defects simply needs to be greater than this. As such $\lambda = 2$ for this decoder.

\subsection{MWM based decoder}

This decoder is based on the graph theoretic problem of finding matchings \cite{hutter:improved}. Though it is an HDRG decoder, it is not greedy in general. Instead it uses techniques that perform optimzation of the correction operator over long ranges.

One such technique is the use of `shortcuts'. These are modifications made to the distances between clusters when any neutral cluster is removed. It is a modification that is also possible for the above decoders, and has been found to allow better decoding \cite{hutter:improved,wootton:simple}. However, for the applicability of the proof, we consider this decoder without the use of shortcuts.

This decoder uses a tunable parameter, $\Lambda$, that can vary between $0$ and $1$. For any given code, system size and noise model, $\Lambda$ can be set at whatever value gives the best results. Proving a threshold for any value of $\Lambda$ therefore proves it also for the decoder in general. The decoder can only be regarded as greedy for $\Lambda=0$, and so we focus on this.

For the case of $\Lambda=0$, this decoder works in a similar way to expanding diamonds, building clusters by pairing existing clusters. Clusters are only ever paired when they are mutual nearest neighbours, i.e. neither has a neighbour closer than the other.

This means that, like expanding diamonds, the largest distance required for all defects within an independent cluster to see each other is its width. In order for none to see any defects outside the independent cluster before they become neutral, the distance to other defects simply needs to be greater than this. As such $\lambda = 2$ for this decoder.

\section{Application to non-Abelian models}

We now consider models for which the syndrome can be interpreted in terms of non-Abelian anyons, such as quantum double models based on a non-Abelian group \cite{double}.

\subsection{Single shot error correction}

Single shot error correction of non-Abelian anyons has previously been studied numerically \cite{wootton:error,brell:14,hutter:improved,burton:15}. These studies provided evidence of a finite threshold, but no formal proof has yet been presented. However, such a proof follows immediately from the discussions of the Abelian case above.

In the single shot case, errors create an anyon configuration. To correct a chunk, the anyons it creates simply need to be fused to annihilate them. In general, moving a non-Abelian anyon requires a controlled operation on all the spins on which its syndrome operator has support. The size of the correction operation will therefore be slightly bigger for the non-Abelian case than the Abelian one to account for this. Specifically, moving the anyons generated by a chunk of width $W$ together requires a correction operator of width at most $W+2$. Greedy HDRG decoders will therefore certainly satisfy Property \ref{prop:1} for this case.

Decoders will also satisfy Property \ref{prop:2} for the single shot non-Abelian case in exactly the same way as for the Abelian. As long as the anyons within each chunk see each other before they see those of other chunks, they will mutually annihilate without affecting or being affected by the anyons of other chunks.

The non-trivial braiding of the non-Abelian anyons will not have any effect on either property. This can be simply seen by the same induction as in Lemma \ref{lemma:2}. Chunks centred around elements of $F_0$ are spatially separated from all others, and so the anyons created by such chunks will not have braided around those of others. Nor will the correction operator that fuses the anyons cause such braiding. The anyons for such such chunks will therefore still annihilate, even if the braiding causes changes in intermediate fusion results. Chunks centred around elements of $F_1$ will similarly remain independent of all remaining errors, and so on for higher levels.

Single shot non-Abelian decoding can therefore be performed by the decoders discussed above, and will have the same values of $\lambda$. The decoding will lead to the threshold noise rates of Eq. (\ref{eq:thresh}).

\subsection{Syndrome for active error correction}

Though there may be many types of anyon possible in any given anyon model, from henceforth we will not distinguish between them for the sake of simplicity. We are therefore only concerned with whether each position in the code is occupied or unoccupied by an anyon according to the measurement results. For the Fibonacci model, in which there is only one non-trivial anyon type, this is the most detailed case possible.

For codes with Abelian anyons, decoding can be postponed until final readout. Furthermore, errors can be corrected effectively by a basis change for the affected spins. This removes the need to physically apply correction operators. Unfortunately, non-Abelian codes share neither of these useful traits. Measurement of anyon occupation alone does not extract sufficient information for good decoding. Fusion of the anyons must also be performed continuously through the process \cite{wootton:error}. These attempted operations performed by the decoder must therefore be taken into account when interpreting the measurement results. Note also that this action will, in general, lead to higher error rates on the spins involved in the anyon transport. However, this can simply be incorporated into the maximum error rate $p_s$ for spins.

Using the measurement results we must construct the syndrome, which will be used as an input for the decoder. For the Abelian case, this was done by assigning a defect to the syndrome lattice wherever there is a change in this measured syndrome. This is because such points are necessarily the endpoints of error chains, and so can be used to determine a likely set of errors that could have caused the measured syndrome. The defects, and hence the anyons, are then removed using this as a guide. The same approach should be taken for the non-Abelian case: points that are necessarily the endpoints of error chains must be identified, and used to remove the anyons.

For the non-Abelian case, it is not true that the measurement result for a given plaquette will only change due to errors. Since the decoder needs to move anyons in order to fuse them, some changes will be expected. For example, if an anyon was measured at $(P,t-1)$ and then moved, it would be expected that no anyon would be measured at $(P,t)$. In fact, it would be unexpected if an anyon was measured at $(P,t)$, and so a lack of change would be the signature of an error in this case. The measurement results at $t-1$, along with the set of movements attempted by the decoder between $t-1$ and $t$, should then be used to determine the expected set of measurement results at $t$. Any point at which the measurement results at $t$ differ from this are a signature of an error, and so should be associated with a defect.

A plaquette $P$ is expected to be empty if it was empty at $t-1$ and no anyon was moved to it, or if it held an anyon at $t-1$ but it was moved away. It is expected to be full if it held an anyon at $t-1$ but no attempt was made to move it, or if it was empty and a single anyon was moved to it. There are two remaining cases, both of which correspond to fusion of anyons: at least one anyon is moved onto a plaquette holding another, or several are moved onto the same plaquette. In these cases there is no expectation either way, since either result could be due to fusion rather than the effects of adjacent errors. As such, no defect will be assigned to such a $(P,t)$.

If no attempts to move the anyons are ever made, a plaquette $P$ will be measured to hold an anyon from the time $t$ at which it was unexpectedly measured for the first time, and the time $t'$ at which it unexpectedly disappeared. A defect will be assigned to both $(P,t)$ and $(P,t')$, since these are necessarily endpoints of error chains. However, note that error chains for the non-Abelian case can terminate anywhere there is an anyon. Such errors chains will not necessarily change the anyon occupancy, and so cannot always be detected by a defect. The `world line' between the defects at $(P,t)$ and $(P,t')$ should therefore also be included in the syndrome given to the decoder. If an anyon is still present at the most recent time slice, the world lines will terminate on these present time anyons rather than defects.

When attempts to move anyons are made, the corresponding world lines should be dragged along with the intended movements. These will combine at fusion events, creating larger `world nets'.

The decoder must find a set of errors that could explain the configuration of defects on the syndrome lattice. In order to do this, error chains can be proposed which connect each defect to another point at which an error chain can end: another defect, a present time anyon or a net.

Any valid error chain corresponds to a proposal for how the missing portions of anyon world nets (corresponding to creation, movement, etc) should be filled in. Once such a chain has been proposed, the corresponding portion can then be added to the world net, and the related defects can be removed from the syndrome lattice. When the error chain connects a defect to a present time anyon, this anyon should also be removed from the syndrome. This is because such an error chain proposes that these apparent anyons are in fact due to measurement errors.

When proposing error chains, one must be careful to determine whether the resulting error net will have portions that are connected to the rest only by a single time-like world line. Such structures would imply that an anyon has been created from vacuum, violating their conservation laws. Such structures must therefore be avoided.

Note that the above is not necessarily true if some part of the structure braids around another world net. This is due to the effects of braiding non-Abelian anyons. However, even in this case, an error chain directly connecting the two nets would correspond to a simpler error, and so could be considered instead. Otherwise, two nets that are braided should be considered to be a single net, just as if they had been connected by an error chain.

When a world net does not contain defects, its neutrality can be considered. If the net contains no present time anyons, the proposed error chain is sufficient to explain the observed defects without otherwise affecting the syndrome. Such a net may therefore be considered to be neutral. If it terminates in multiple present time anyons, it could be that fusion of these will lead to annihilation, and hence neutrality. These anyons should therefore be moved together by the decoder to determine whether this is indeed the case. If the net terminates in a single present time anyon, at least one further error chain is required to explain the presence of this anyon. Such world nets are therefore not neutral, and should be considered to be defects themselves. Note that these defects can be paired with their own present time anyon, reflecting the possibility of a chain of measurement errors between the time of fusion and the present time, as well as being paired with other defects, nets or present time anyons.

For an error chain connecting a defect with a present time anyon, the removal of the defect from the syndrome should not be considered permanent. This is because further timeslices might show that this proposed error chain was, in fact, unlikely. Any such defects should be reinstated after the movement performed at each timeslice, along with the rest of their net.

\begin{figure}[t]
\begin{center}
{\includegraphics[width=8cm]{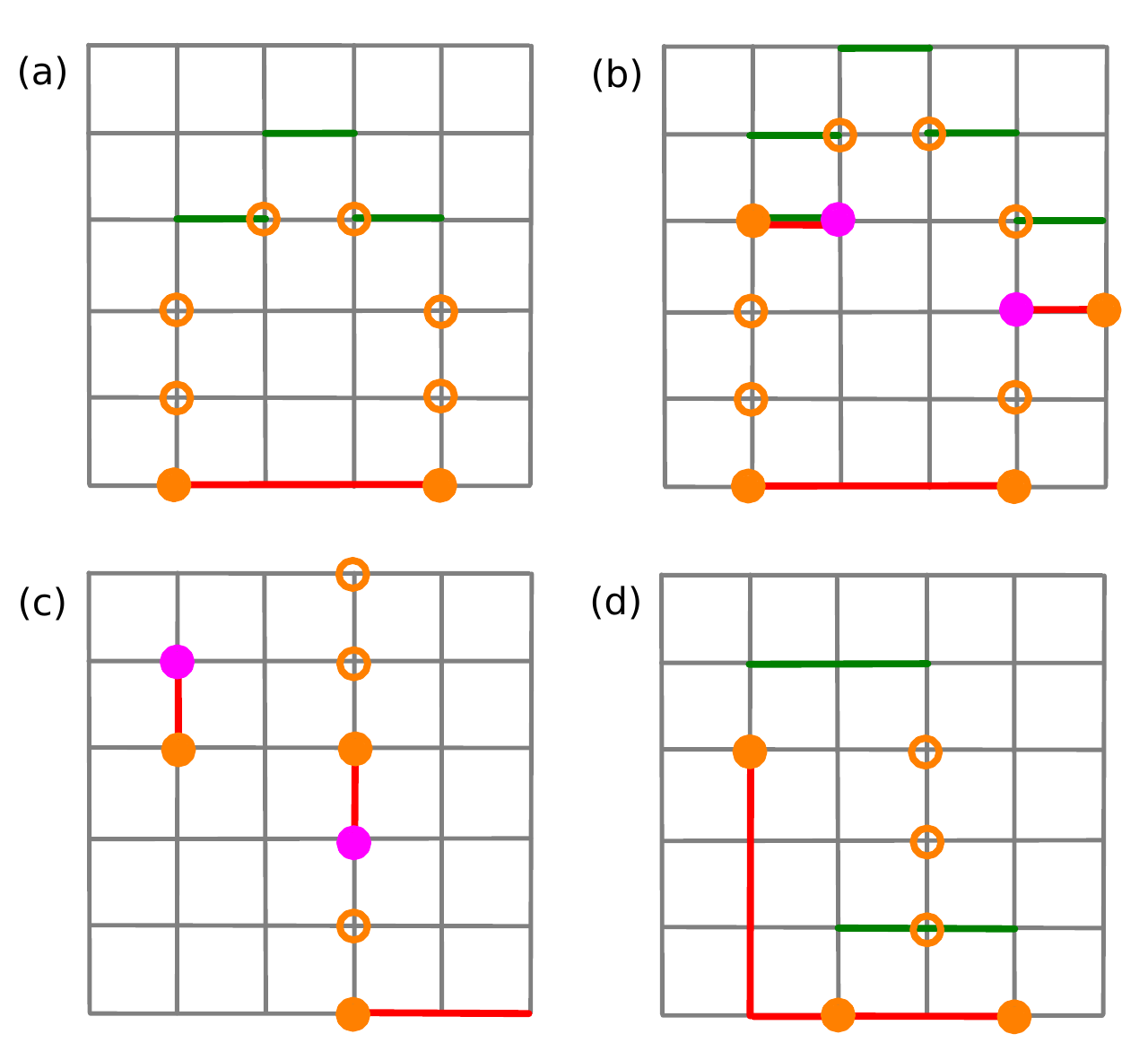}}
\caption{\label{fig:errors} Examples of errors, and how these may be dealt with by a decoder. Here only a one dimensional slice of the spatial plane is considered, represented horizontally. Time corresponds to the vertical direction, with time flowing in the upwards direction. Points at which anyons are measured are denoted by orange circles. These are filled if the anyon is unexpected, and so correspond to a defect. Points at which no anyon was measured when one was expected, another defect, are denoted by pink circles. Red lines denote errors (horizontal for spin and vertical for measurement) while green lines denote attempted anyon movement by the decoder. (a) A chain of three spin errors creates a pair of anyons. For a few time steps it is most likely that the anyons are simply due to measurement errors. After this they are most likely to have been created in a pair, and so are moved towards each other until annihilation. (b) Same as before, except that an additional spin error moves the right anyon. This creates an additional pair of defects. The first movement operation applied to the left anyon also fails, meaning that the anyon does not move as expected. This also creates an additional pair of defects. Nevertheless, these effects are accounted for, and the anyons are finally annihilated. (c) Two examples of measurement errors, one where no anyon is expected and one where one is. Both lead to a pair of defects. (d) A chain of spin errors create three anyons, though one is hidden for a time by measurement errors. The decoder first pairs those it can see, but they do not annihilate. The final anyon, once visible, is annihilated with the fusion product.}
\end{center}
\end{figure}

\subsection{Decoding for active error correction}

The syndrome for non-Abelian anyon models above is largely the same as that for Abelian ones. The main difference is the need to consider world lines at which error chains can end in an undetected manner. This difference does not prevent the greedy HDRG decoders discussed above from being straightforwardly applied to the non-Abelian case. As such, one might expect that the threshold proof straightforwardly applies also.

Unfortunately, this is not the case. The differences between the Abelian and non-Abelian decoding problems prevent any decoder from satisfying  Property \ref{prop:2}, even in general. This prevents the application of the proof, and demonstrates a significant difference between the decoding of Abelian and non-Abelian models.

To see why Property \ref{prop:2} does not hold, consider a chain of spin errors of length $l$, which create a pair of anyons located at the endpoints. Annihilating these requires moving them together. Assuming that each anyon can only be moved to a neighbouring plaquette in each time step, this means that the anyons will still be present for at least a time $l/2$ after their creation. If a single error occurs adjacent to one, it can cause it to move. This error is therefore certainly not independent of the string. However, the width of the single error is $W=1$, and the distance from it to any other error is can be up to $l/2$. Its lack of independence therefore implies $l/2 < \lambda/2$ for any $l = \Theta(L)$. As such $\lambda = \Theta(L)$, which contradicts Property \ref{prop:2}.

To remove this effect, one could assume that a non-Abelian decoder can move anyons arbitrarily far in each time step. Though an unphysical assumption, it could be used to make progress towards a threshold proof. Unfortunately, even this is not enough. Consider again the above chain of errors. If measurement errors are as likely as spin ones, until a time $l/2$ has passed it is more likely that the anyons are a result of measurement errors than the chain of $l$ spin errors. Proposed error chains will therefore pair them with their present time anyons, and so they will not be moved. They will therefore still present for at least a time $l/2$ after their creation, and so the above arguments apply even when anyons can be moved arbitrarily quickly.

Note that these issues do not imply a lack of threshold for error correction of non-Abelian anyons. Instead it simply shows that chunks that would be independent in Abelian codes can still interact with each other if the code is non-Abelian. However, this still requires them to be sufficiently close. A logical error would therefore require a collection of chunks that would cause these effects to percolate across the lattice. This percolation is likely to be highly suppressed for low enough $p$, and so a threshold proof for the non-Abelian case is likely to be possible. However, such a proof of the threshold theorem for non-Abelian anyons is yet to be found.

\section{Conclusions}

Here we have provided a very general threshold proof, applicable to both Abelian and non-Abelian decoding problems for general anyon models. The proof also applies to a general class of decoders. However, the threshold theorem for active error correction of non-Abelian anyons still remains to be proven. We have made contributions in this direction, by studying how active error correction may be performed in this case. Numerical and analytical verification of a threshold, either for general models and decoders or for specific cases, is left to future work.

\section{Acknowledgements}

JRW would like to thank Fern Watson for discussions of the proof in \cite{bravyi:13}. The authors acknowledge the SNF and QSIT for support.

\bibliography{refs}

\end{document}